\documentclass[preprint,12pt]{elsarticle}

\usepackage{amsmath,amsfonts,amsthm,amssymb,paralist,subfigure,graphicx,amsbsy,float,epsfig,color}
\usepackage{cuted,mathtools,lipsum}

\usepackage[ruled,vlined,linesnumbered]{algorithm2e}

\usepackage{stmaryrd,url}

\usepackage{amsthm}
\newtheorem{lem}{Lemma}
\newtheorem{ass}{Assumption}
\newtheorem{thm}{Theorem}

\newtheorem{rem}{Remark}

\def\mb{\mathbf}

\def\mc{\mathcal}

\journal{Digital Signal Processing}

\begin{document}

\begin{frontmatter}

\title{Distributed Automatic Generation Control subject to Ramp-Rate-Limits: Anytime Feasibility and Uniform Network-Connectivity
}

\author[Sem]{Mohammadreza Doostmohammadian}
\affiliation[Sem]{Faculty of Mechanical Engineering, Semnan University, Semnan, Iran, doost@semnan.ac.ir.}
\author[AA]{Hamid R. Rabiee} 
\affiliation[AA]{Computer Engineering Department, Sharif  University of Technology, Tehran, Iran rabiee@sharif.edu.}

\begin{abstract}
	This paper considers automatic generation control over an information-sharing network of communicating generators as a multi-agent system. The optimization solution is distributed among the agents based on information consensus algorithms, while addressing the generators' ramp-rate-limits (RRL). This is typically ignored in the existing linear/nonlinear optimization solutions but they exist in real-time power generation scenarios. Without addressing the RRL, the generators cannot follow the assigned rate of generating power by the optimization algorithm; therefore, the existing solutions may not necessarily converge to the exact optimal cost or may lose feasibility in practice. The proposed solution in this work addresses the ramp-rate-limit constraint along with the box constraint (limits on the generated powers) and the coupling-constraint (generation-demand balance) at all iteration times of the algorithm. The latter is referred to as the anytime feasibility and implies that at every termination point of the algorithm, the balance between the demand and generated power holds. To improve the convergence rate of the algorithm we further consider internal signum-based nonlinearity. We also show that our solution can tolerate communication link removal. This follows from the uniform-connectivity assumption on the communication network. 
\end{abstract}

\begin{graphicalabstract}
	\includegraphics{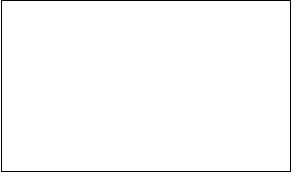}
\end{graphicalabstract}

\begin{highlights}
	\item Clearly defining the problem of distributed automatic generation control via multi-agent systems and presenting a formal mathematical optimization framework that captures its essence.
	\item Proposing a distributed setup for optimal convergence and analyzing its different features in terms of feasibility, convergence, and network connectivity 
	\item Addressing real-world constraints in terms of ramp-rate-limits and model nonlinearities in the optimization problem 
\end{highlights}


\end{frontmatter}

\section{Introduction} \label{sec_intro}
As our power systems become increasingly complex and interconnected, the need for efficient and reliable automatic generation control (AGC) strategies becomes crucial. AGC maintains the balance between power generation and demand, ensuring grid stability and optimal performance. However, integrating renewable energy sources and the growing reliance on distributed generation have introduced new challenges, including ramp-rate-limits (RRL) on power generation devices. Ramp-rate limits are restrictions imposed on the rate at which power generation can be increased or decreased. Respecting these limits ensures the stability of the power system by preventing sudden and large fluctuations in generation levels. Motivated by recent advances in Cloud-computing and distributed and parallel processing along with embedded low-cost CPUs and wireless communications, coordinated control and optimization algorithms among the power generating nodes are proposed which are essential for effective \textit{distributed} AGC. The design of these distributed optimization algorithms must meet the limitations imposed by the ramp-rate constraints. Otherwise, it may cause instability and service disruption. This work proposes effective and reliable distributed AGC solutions that ensure grid stability while adhering to RRL constraints. Moreover, our distributed setup allows to address different nonlinearity on the node dynamics and links, which is an improvement over the existing works.

Distributed optimization \cite{abboud2015distributed} and resource allocation \cite{zhang2020dynamic,cherukuri2015distributed,doan2020distributed,ding2021differentially,xu2017distributed} are prevalent in a wide range of applications from power grid \cite{kar2012distributed,yang2013consensus,chen2016distributed} and CPU scheduling \cite{OJCSYS,grammenos2023cpu} to machine learning \cite{Yan2024killing,qureshi2020s}, power resource allocation \cite{SHI2020102850}, and information consensus \cite{fereydounian2023provably,yu2024consensus}. Distributed algorithms are more robust to single-node-of-failure and utilize the possibility of parallelizing the computations/processing over a multi-agent network.
Recently relevant distributed algorithms are proposed to solve general resource allocation problems over multi-agent networks \cite{xiao2006optimal,lin2020predefined,wu2023distributed}. Among possible applications in electric power systems, the Economic Dispatch Problem (EDP) and Automatic Generation Control (AGC) found more attention in the literature. AGC is a technique used to maintain the balance between power generation and power demand \cite{miller1994power,glover2012power}. It is an essential control mechanism that prevents load-demand deviations, which can lead to disruptions, equipment damage, and even blackout events and also helps to optimize the generation dispatch and scheduling, reduce energy imbalances, and support the integration of renewable energy sources, which often have variable and intermittent generation characteristics.
AGC checks the generation-load balance by constantly measuring the system frequency to remain within a desired range\footnote{The process of measuring the system frequency is not the focus of this paper. In this work, we aim to adjust the generation powers via distributed algorithms while meeting the power-demand constraint.}. In this paper, we are interested in distributed optimization algorithms to adjust the AGC power mismatch. This is closely tied with the EDP problem \cite{li2015connecting,boqiang2009review,chowdhury1990review}. The goal of EDP is to minimize the total operating costs of a group of generators by
determining the most economical way to allocate power generation among them to meet a given load demand \cite{kar2012distributed,elsayed2014fully}.
Generators are incurred with different costs to produce electrical energy and to transmit the energy to the load. The EDP algorithm runs every few minutes to adjust the optimal combination of power-generating units to meet the electricity demand at any given time, while minimizing the total cost of generation, which includes fuel costs, maintenance costs, and other operational costs \cite{handbook}.

AGC can be implemented at different levels, such as local control within individual generating units (distributed solution), or at a central control level that coordinates the operation of multiple generating units across the power grid (centralized solution). The distributed solutions, mainly built based on information consensus algorithms, are (i) resilient to single-node-of-failure, (ii) with distributed computation among a set of processing nodes, and (iii) compatibility with parallel processing and decentralized algorithms (allowing for cloud-computing and distributed data processing). In this paper, we consider a distributed AGC setup to solve the associated distributed optimization problem while meeting the coupling equality constraint addressing the power-demand balance at all times. This \textit{all-time feasibility} implies that at any termination point of the algorithm, the solution is feasible (meets the constraint). This is in contrast to existing ADMM-based solutions which asymptotically meet this constraint at the optimal convergence of the algorithm and not along the solution. Some ADMM-based resource allocation models include: constraint-coupled optimization \cite{falsone2020tracking,banjac2019decentralized,carli2019distributed}, parallel optimization \cite{cdc22,cdc_wei}, passivity-based solution \cite{notarnicola2022passivity}, and consensus-based methods \cite{li2018admm,wang2021admm,he2019optimizing}. The main feature of the proposed solution in this work is to address the so-called ramp-rate-limits (RRLs) associated with the generators in the AGC problem \cite{d2022ramp,venkat2008distributed}. Note that the AGC generators’ deviations are subject to limits based on the available power reserves (this is referred to as the box constraints) and
also on RRLs (or rate saturation). RRL imply that
the speed/rate at which the produced power by generators can increase/decrease is
constrained by a limit and the generators cannot abruptly change their generated powers at any (high) rate. This nonlinear constraint implies that a linear (and ideal) solution in theory, without RRL consideration, may not necessarily converge. This is because the real generators cannot follow the linear rate assigned by the theory, but need to be assigned with a saturated rate of decrease or increase. Without proper power assignment, the real network of generators either fails the generation-demand balancing constraint (resulting in service disruption) or loses cost-optimality. In this paper, we propose solutions to address such RRLs on generators. This is the main feature of the solution proposed in this work which distinguishes our solution from the existing linear \cite{cherukuri2015distributed,doan2020distributed,kar2012distributed} and nonlinear \cite{yang2013consensus,chen2016distributed} methods. Moreover, such a nonlinear constraint on the ramp rates cannot be addressed via ADMM-based solutions in \cite{falsone2020tracking,cdc22,cdc_wei,banjac2019decentralized,notarnicola2022passivity,li2018admm,carli2019distributed}.
The other features of the proposed solution are: (i) it works over uniformly connected networks in contrast to \cite{falsone2020tracking,cdc_wei,banjac2019decentralized}, and (ii) it can address non-quadratic objective functions with additive penalty terms due to the box constraints in contrast to \cite{kar2012distributed}. The uniform-connectivity assumption further implies that our solution is resilient to link failure, as shown in theory and by simulation in this paper.

\textit{Paper organization:} Section~\ref{sec_prob} formulates the problem and states the preliminaries. The main results are given in Section~\ref{sec_sol} proposing the solution subject to RRL. The simulations are presented in Section~\ref{sec_sim} and Section~\ref{sec_con} concludes the paper.

\section{Problem Formulation} \label{sec_prob}
\subsection{The Optimization Problem}
The AGC regulates the generated power based on some pre-determined reserve limits to compensate for any mismatch between generation and generation-demand. Assuming that the power mismatch (due to, e.g., failed generators) is known, the optimization problem is to allocate the power mismatch to other generators while optimizing the power deviation costs (see Fig.~\ref{fig_agc}).
\begin{figure}[]
	\centering
	\includegraphics[width=3.5in]{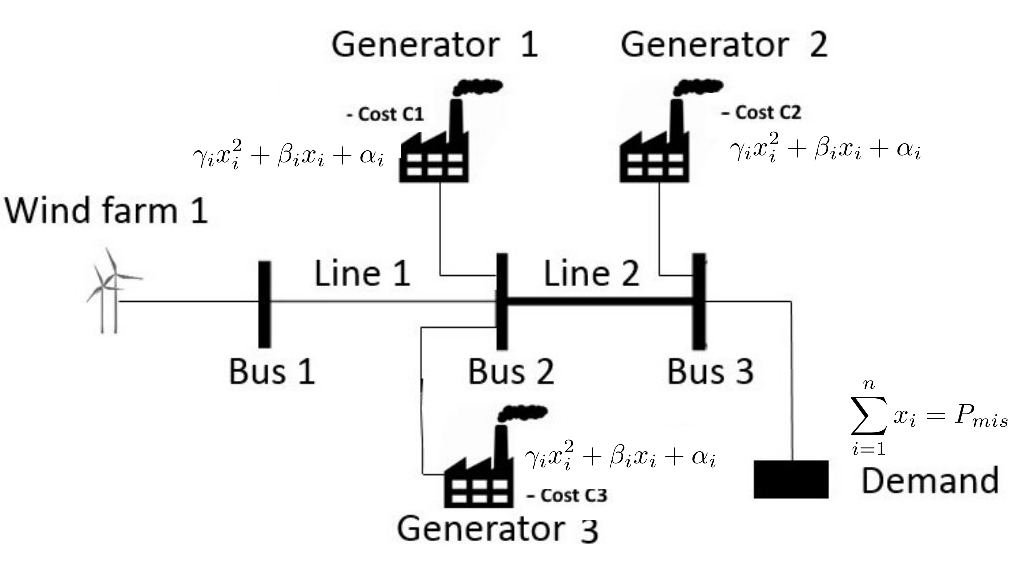}
	\caption{The automatic generation control problem aims to schedule the assigned power to the generators to adjust the power demand mismatch while optimizing the power generation costs.
	}
	\label{fig_agc}
	\vspace{-0.5cm}
\end{figure}
The optimization problem is then in the following form \cite{chen2016distributed,kar2012distributed}:
\begin{align} \label{eq_f_quad}
	\min_\mb{x} &\sum_{i=1}^n \overline{f}_i(x_i) = \min_\mb{x} \sum_{i=1}^n \gamma_i x_i^2+ \beta_i x_i + \alpha_i,\\ \nonumber ~\mbox{s.t.}~&\sum_{i=1}^n x_i = P_{mis}, ~~ -\underline{m}_i \leq x_i \leq \overline{m}_i
\end{align}
where $P_{mis}$ is the power mismatch, $x_i$ is the power state at the generator $i$, $\underline{m}_i,\overline{m}_i$ denote the so-called box constraints (upper and lower limit of generated power at the generator $i$). There is an RRL constraint on the solution
\begin{align} \label{eq_rrl}
	-R \leq \dot{x}_i \leq R
\end{align}
with $\dot{x}_i$ as the rate of change in the generator power and $R$ as the RRL. The parameters $\gamma_i,\beta_i,\alpha_i$ are the parameters of the cost function. Examples of these parameters are given in Table~\ref{tab_par} \cite{wood2013power,yang2013consensus}, and the minimum generated power parameter $\underline{m}_i$ is $20MW$ at these example generators.
\begin{table} [hbpt!]
	\centering
	\caption{Power generation cost parameters based on generator types \cite{wood2013power,yang2013consensus}}
	\begin{tabular}{|c|c|c|c|c|}
		\hline
		Type & $\alpha_i$ (\$/$h$) & $\beta_i$ (\$/$MWh$) & $\gamma_i$ (\$/$MW^2h$) & $\overline{m}_i$ ($MW$)\\
		\hline
		A & 561& 2.0 & 0.04 & 80\\
		\hline
		B& 310& 3.0& 0.03 & 90\\
		\hline
		C &78 & 4.0& 0.035 & 70\\
		\hline
		D& 561& 4.0& 0.03 & 70\\
		\hline
		E &78 & 2.5& 0.04 & 80\\
		\hline		
	\end{tabular}
	\label{tab_par}
\end{table}

\subsection{Preliminaries}
The network/graph $\mc{G}$ over which the generator units communicate is modelled as an undirected network. The network may lose connectivity due to packet drop or link removal. This implies that the network is volatile and can change over time. Therefore, for long-term connectivity, the network needs to be uniformly connected over every certain time-interval $B$ (or B-connected), i.e., $\cup_{t}^{t+B} \mc{G}(t)$ is connected.

\begin{rem}The condition on uniform connectivity requires that, although the network topology may vary over time and some communication links could intermittently fail, the union of the graphs over every bounded time interval of length 
	$B$ is connected. In other words, if one considers the combined network formed by aggregating all links active during any time interval $[t,B+t]$, this combined graph must be connected. This assumption allows for temporary network disconnections due to link failures or packet drops which are common in practical communication scenarios. However, it ensures that information can still propagate across the entire network over time, enabling all nodes to collaboratively converge to a global consensus. Convergence analysis under uniform-connectivity is typically based on the fact that repeated joint connectivity ensures delivery of required information throughout the network, while the bounded interval $B$ sets the maximum period of possible disconnection. This is key for distributed AGC since each generator node only has partial knowledge (local information on the constrained optimization objective) and relies on neighbour communications to update its state based on cost optimization or power imbalance/mismatch.
\end{rem}

\begin{ass} \label{ass_W}
	The adjacency matrix $W$ is symmetric (not necessarily stochastic) and its associated network of generators is uniformly B-connected.
\end{ass}

Entry $W_{ij}>0$ denotes the weight of the link $(i,j)$ between node $i$ and $j$ and $W_{ij}=0$ implies no link. Define matrix $L$ as the Laplacian matrix associated with $W$, i.e., $L=D-W$ with $D$ defined as a diagonal matrix in the form $D = \mbox{diag}[\sum_{j=1}^n W_{ij}]$.

Next, we review some useful lemmas from convex optimization and algebraic theory. The following lemma can be found in the standard optimization books, e.g., \cite{boyd2004convex}.
\begin{lem} \label{lem_strict}
	For a strictly-convex function $F:\mathbb{R}^n \mapsto \mathbb{R}$ with locally Lipschitz derivatives such that $2 v <  \frac{d^2 f_i(x_i)}{dx_i^2} < 2 u$, points $\mb{x}_1, \mb{x}_2 \in \mathbb{R}^n$, and $\delta := \mb{x}_1-\mb{x}_2$,
	\begin{align} \label{eq_taylor_1}
		F(\mb{x}_1) \geq F(\mb{x}_2) + \nabla F(\mb{x}_2)^\top \delta +  v\delta^\top \delta
		\\
		F(\mb{x}_1) \leq F(\mb{x}_2) + \nabla F(\mb{x}_2)^\top \delta  +  u\delta^\top \delta
		\label{eq_taylor_2}
	\end{align}
\end{lem}

Note that although the cost in \eqref{eq_f_quad} is quadratic, adding some penalty terms (discussed later) makes it non-quadratic, and therefore, general terms for the cost function are considered.

\begin{lem} \label{lem_sum}
	Given an odd and sign-preserving nonlinear mapping $g:\mathbb{R} \mapsto \mathbb{R}$, under Assumption~\ref{ass_W}, for $\phi \in \mathbb{R}^n$,
	\begin{align} \label{eq_sum_lem}
		\sum_{i=1}^n \phi_i \sum_{j=1}^n W_{ij} g(\phi_j-\phi_i) = \sum_{i,j=1}^n \frac{W_{ij}}{2} (\phi_j-\phi_i) g(\phi_j-\phi_i)
	\end{align}
\end{lem}
\begin{proof}
	From Assumption~\ref{ass_W} for every $i,j$, we have $W_{ij} = W_{ji}$. Also for $x \in \mathbb{R}$ we have $g(x) = -g(-x)$. Thus,
	\begin{align} \nonumber
		\phi_i W_{ij} g(\phi_i-\phi_j) + &\phi_j W_{ji}  g(\phi_i-\phi_j) \\ \nonumber
		& = W_{ij}(\phi_i-\phi_j) g(\phi_j-\phi_i) \\
		& = -W_{ji}(\phi_i-\phi_j) g(\phi_i-\phi_j).
	\end{align}
	and the proof follows.
\end{proof}
\section{Gradient-Laplacian Tracking based on Saturation Function} \label{sec_sol}
The basic idea behind gradient-based distributed optimization is to leverage the concept of gradient descent, a popular optimization algorithm that iteratively updates the model parameters in the direction of the steepest descent of the objective function. In distributed optimization, instead of computing the gradients and updating the parameters on a single node, these tasks are distributed among multiple nodes, and the nodes collaborate to find the optimal solution collectively based on their local objective functions. The idea is that each node locally computes the gradients of its objective function using its local data and model parameters and shares the gradient information over an information-sharing network.

\subsection{Addressing RRLs}
We propose the following $1$st-order gradient-Laplacian solution subject to RRL $-R \leq \dot{x}_i \leq R$ to solve distributed problem \eqref{eq_f_quad}:
\begin{align}
	\dot{x}_i = -\eta \frac{R}{W_{max}} \sum_{j \in \mc{N}_i} W_{ij} \mbox{sat} \Big(\nabla f_i(x_i) -  \nabla f_j(x_j)\Big)
	\label{eq_sol}
\end{align}
with $\mc{N}_i$ as the neighbors of the generator $i$. It is clear that the sat function keeps $\dot{x}_i$ within the RRLs. In discrete-time we have,
\begin{align} \nonumber
	{x}_i&(k+1) = {x}_i(k) \\
	&- \eta \frac{R}{W_{max}} \sum_{j \in \mc{N}_i} W_{ij} \mbox{sat} \Big(\nabla f_i(x_i(k)) -  \nabla f_j(x_j(k))\Big)
	\label{eq_sol_d}
\end{align}
with $k$ as the discrete time-index, $0<\eta\leq1$ as the step-rate, $W_{ij}$ as the link weight, $\nabla f_i(x_i)$ as the gradient of cost $\overline{f}_i$ plus smooth penalty terms to address the box constraints, $W_{max}= \max_{1\leq i \leq n} \sum_{j =1}^n W_{ij}$, and $\mbox{sat}$ denotes the nonlinear saturation function. Note that the smooth penalty terms are to be added to the cost function as the new objectives are in the forms
\begin{align}
	f_i = \overline{f}_i + c([x_i - \overline{m}_i]^+ + [\underline{m}_i - x_i ]^+)
	\label{eq_f}
\end{align}
with
\begin{align}
	[u]^+=\max \{u, 0\}^\sigma,~\sigma \in \mathbb{N}
	\label{eq_sigma}
\end{align}
and weighting constant $c \in \mathbb{R}^+$. In order to have smooth penalty terms in the objective function it is typical to consider $\sigma \geq 2$. The name gradient-Laplacian comes from the linear form of \eqref{eq_sol} which is in the following form:
\begin{align} \nonumber
	\dot{x}_i &= -\eta  \sum_{j \in \mc{N}_i} W_{ij} \Big(\nabla f_i(x_i) -  \nabla f_j(x_j)\Big) \\
	&= -\eta L \nabla F(\mb{x}) \label{eq_lin}
\end{align}
where $\nabla F(\mb{x}) := (\nabla f_1(x_1),\dots, \nabla f_n(x_n))$. Note that this linear solution (proposed by \cite{cherukuri2015distributed,kar2012distributed,doan2020distributed}) does not necessarily address the RRL of generators $-R \leq \dot{x}_i \leq R$. Recall that primal-dual-based solutions, e.g., ADMM resource allocation \cite{falsone2020tracking,cdc22,cdc_wei,banjac2019decentralized,notarnicola2022passivity,li2018admm,wang2021admm,carli2019distributed}, cannot address this RRL constraint and therefore, in reality, the generators' rates cannot exactly follow the rates assigned by the algorithm. This implies that their solution in the presence of RRLs may not necessarily converge to the exact optimal cost or it may lose the generation-demand feasibility. This is the motivation behind considering the gradient-Laplacian solution instead of ADMM solutions.

\begin{lem}
	Under Assumption~\ref{ass_W}, the solution \eqref{eq_sol} is all-time feasible.
\end{lem}
\begin{proof}
	Note that the sat function is an odd and sign-preserving nonlinear mapping. Then, following the symmetry of $W$ from Assumption~\ref{ass_W} we have
	\begin{align} \nonumber
		\sum_{i=1}^n \dot{x}_i &= -\sum_{i=1}^n \frac{\eta R}{W_{max}} \sum_{j \in \mc{N}_i} W_{ij} \mbox{sat} \Big(\nabla f_i(x_i) -  \nabla f_j(x_j)\Big) \\ \nonumber
		&= 0
	\end{align}
	Similarly for discrete-time case we have $\sum_{i=1}^n x_i(k+1)= \sum_{i=1}^n x_i(k) = P_{mis}$ implying that the solution is feasible at all times.
\end{proof}

\begin{lem} \label{lem_tree}
	Let $x_i^*$s denote the steady-state equilibrium under dynamics \eqref{eq_sol} or \eqref{eq_sol_d}. Under Assumptions~\ref{ass_W}, $(\nabla f_1(x_1^*),\dots,\nabla f_n(x_n^*)) \in \mbox{span}(\mb{1}_n)$ and this equilibrium matches the optimal point of problem~\eqref{eq_f_quad}.
\end{lem}
\begin{proof}
	The proof follows similar to the proof of Lemma~2 and Lemma~7 in \cite{OJCSYS}.
\end{proof}

\begin{thm} \label{thm_converg}
	Starting from initial condition such that $\sum_{i=1}^n x_i(0) = P_{mis}$
	and under Assumptions~\ref{ass_W}, dynamics \eqref{eq_sol} or \eqref{eq_sol_d} converges to the optimal solution of  $\mb{x}^*:=({x}_1^*,\dots,{x}_n^*)$ described in Lemma~\ref{lem_tree} for sufficiently small $\eta$.
\end{thm}
\begin{proof}
	Define the Lyapunov-type function $\overline{F} := \sum_{i=1}^n f_i(x_i)-\sum_{i=1}^n f_i(x^*_i)$ also known as the \textit{residual cost}, i.e., subtracting the optimal cost from the overall cost. Based on the Lyapunov theory, we prove $\dot{\overline{F}} \leq 0$ under continuous-time solution~\eqref{eq_sol} and only $\dot{\overline{F}}(\mb{x}^*)=0$.
	For continuous-time cases,
	
	\small
	\begin{align} \nonumber
		\dot{\overline{F}} &= \nabla F^\top \dot{\mb{x}} \\ \nonumber
		&= -\sum_{i =1}^n \nabla f_i(x_i) \eta \frac{R}{W_{max}} \sum_{j =1}^n W_{ij} \mbox{sat} \Big(\nabla f_i(x_i) -  \nabla f_j(x_j)\Big) \\
		&= \sum_{i,j =1}^n \frac{\eta R  W_{ij}}{2W_{max}} (\nabla f_j(x_j) - \nabla f_i(x_i))   \mbox{sat} \Big(\nabla f_i(x_i) -  \nabla f_j(x_j)\Big) \label{eq_proof1}
	\end{align} \normalsize
	where the last equality follows from Assumption~\ref{ass_W}, Lemma~\ref{lem_sum}, and the fact that $\mbox{sat}(\cdot)$ is an odd and sign-preserving function. Eq.~\eqref{eq_proof1} implies $\dot{\overline{F}} \leq 0$, i.e., the Lyapunov function is negative semi-definite where the equality holds at $\nabla f_i(x_i) = \nabla f_j(x_j)$ which represents the optimal equilibrium from Lemma~\ref{lem_tree}. Recall that  Lemma~\ref{lem_tree} implies that at the optimal solution $\nabla f_i(x^*_i) = \nabla f_j(x^*_j)$ which also from Eq.~\eqref{eq_proof1} represents the equilibrium  of the gradient-Laplacian solution~\eqref{eq_sol} (i.e., the  invariant set of the dynamics~\eqref{eq_sol}). Then from LaSalle's invariance principle, the proof follows.
	The proof of the discrete-time solution~\eqref{eq_sol_d} similarly follows by showing $\overline{F}(k+1) < \overline{F}(k)$ from Lemma~\ref{lem_strict}.
\end{proof}

\subsection{Improving the Convergence Rate via Signum Nonlinearity} \label{subsec_sign}
It is known that signum-based nonlinearity can improve the convergence rate of the optimization and consensus problems \cite{parsegov2013fixed,taes,wang2010finite,chen2016distributed}. These are motivated by finite-time and fixed-time gradient flows \cite{haddad2008finite}. Therefore, we add such a nonlinearity to compensate for the low convergence rate of the saturation-based solution. The updated protocol is in the following continuous-time form:
\begin{align}
	\dot{x}_i = -\eta \frac{R}{W_{max}} \sum_{j \in \mc{N}_i} W_{ij} \mbox{sat} \Big(g\Big(\nabla f_i(x_i) -  \nabla f_j(x_j)\Big)\Big)
	\label{eq_sol2}
\end{align}
and in discrete-time,
\begin{align} \nonumber
	{x}_i&(k+1) = {x}_i(k) \\
	&-  \frac{\eta R}{W_{max}} \sum_{j \in \mc{N}_i} W_{ij} \mbox{sat} \Big(g\Big(\nabla f_i(x_i(k)) -  \nabla f_j(x_j(k))\Big)\Big)
	\label{eq_sol_d2}
\end{align}
with $g(u) = \mbox{sgn}^{\mu}(u)$ where  $\mbox{sgn}^\mu(u)=\frac{u^{\mu}}{|u|}$ and  $0<\mu<1$. From these definitions, $\mbox{sgn}^{\mu}(u)$ improves the convergence rate inside the saturation limits.

\begin{thm} \label{thm_converg2}
	Starting from feasible initial conditions $\sum_{i=1}^n x_i(0)= P_{mis}$
	and under Assumptions~\ref{ass_W},  dynamics \eqref{eq_sol2} or \eqref{eq_sol_d2}  is (i) all-time feasible, and (ii) converges to the same optimal solution of  $\mb{x}^*$ described in Lemma~\ref{lem_tree} and Theorem~\ref{thm_converg} for sufficiently small $\eta$.
\end{thm}
\begin{proof}
	For the proof of (i), note that $g(\cdot)$ is an odd and sign-preserving function; therefore, similar to the proof of Lemma~\ref{lem_sum}, we have $\sum_{i=1}^n \dot{x}_i = 0$ and $\sum_{i=1}^n x_i(k+1)= \sum_{i=1}^n x_i(k) = P_{mis}$. This proves all-time feasibility. To prove (ii), the same positive-definite Lyapunov function $\overline{F}$ as in Theorem~\ref{thm_converg}, has negative semi-definite derivative  $\dot{\overline{F}} \leq 0$ under continuous-time solution~\eqref{eq_sol2} and only $\dot{\overline{F}}(\mb{x}^*)=0$. This is because the function $\mbox{sat}(g(\cdot))$ (and $\mbox{sgn}(\cdot)$ function) is odd and sign-preserving similar to Eq.~\eqref{eq_proof1}. In mathematical form\footnote{Note that when the sign nonlinearity is added, Eq.~\eqref{eq_sol2} becomes a discontinuous differential equation. Hence, tools from non-smooth	analysis \cite{cortes2008discontinuous} can be used and the proof follows the same line of reasoning as in the proof of Theorem~\ref{thm_converg}.},
	
	\begin{align} \nonumber
		\dot{\overline{F}}
		= \sum_{i,j =1}^n \frac{\eta R  W_{ij}}{2W_{max}} &(\nabla f_j(x_j) - \nabla f_i(x_i))  \\ &\mbox{sat}\Big(g\Big(\nabla f_i(x_i) -  \nabla f_j(x_j)\Big)\Big)  \label{eq_proof2}
	\end{align} \normalsize
	and 	$\dot{\overline{F}} \leq 0$ following from Assumption~\ref{ass_W} and Lemma~\ref{lem_sum}.  Note that the equality is at $\nabla f_i(x_i) = \nabla f_j(x_j)$, i.e., optimal point from Lemma~\ref{lem_tree}. The proof of discrete-time solution~\eqref{eq_sol_d2} is similarly straightforward from Lemma~\ref{lem_strict}.
\end{proof}

One issue with the additive signum-based nonlinearity is that it introduces unwanted \textit{chattering} to the solution. This chattering refers to the undesirable oscillation of power states at the steady state. This is because of the non-Lipschitz nature of the signum function and can be reduced by decreasing the step-rate $\eta$\footnote{In general, there is a trade-off between the convergence rate and chattering amplitude depending on $\eta$. To avoid possible chattering at the equilibrium, one can replace the traditional sign function with soft sign function used in \cite{10153671}.}. This is a typical issue in fixed-time and finite-time optimization and consensus solutions.
Since the sign function is not differentiable, one approach is to replace it with a smooth approximation that approximates its behaviour. Smooth approximations, such as the sigmoid function or the hyperbolic tangent function, can be used as surrogates for the sign function. This is common in information consensus algorithms and in sliding mode control \cite{nonlin}. By replacing the sign function with a smooth approximation, one can leverage traditional gradient-based optimization methods that rely on differentiability. However, in this case, one cannot theoretically guarantee finite/fixed-time convergence.

The discrete-time version of our distributed solution is summarized in Algorithm~\ref{alg_1}.
\begin{algorithm}
	\textbf{Input:}  $W_{max}$, $\mc{G}$, $W$, $\eta$, $\overline{m}_i$, $\underline{m}_i$, $P_{mis}$, $R$, $f_i(\cdot)$\;
	\textbf{Initialization:} set $k=0$, every generator node $i$ chooses a feasible initial condition, e.g. $x_i(0)=\frac{P_{mis}}{n}$\;
	\While{termination criteria NOT true}{
		Generator $i$ receives a packet including $\nabla f_j(x_j)$ from neighboring generator $j$ over network $\mc{G}$\;
		Generator $i$ computes Eq.~\eqref{eq_sol_d} or Eq.~\eqref{eq_sol_d2}\;
		Generator $i$ shares $\nabla f_i(x_i)$ with neighbors $j$ over network $\mc{G}$\;
		$k \leftarrow k+1$\;
	}
	\textbf{Return} Final state $x_i$ and cost $f_i(x_i)$\;	
	\caption{Generator coordination algorithm for AGC}
	\label{alg_1}
\end{algorithm}

\subsection{Convergence in the Presence of Link Failure}
Link failure in distributed optimization networks refers to the situation where the information-sharing link between two nodes or components in the network becomes unavailable or unreliable. This can occur due to various reasons such as physical cable damage, hardware failure, network congestion, interference, transmission errors, or even packet drops. Some main root causes of the link failure and packet drops in real network environments are as follows: (i) network congestion due to insufficient bandwidth allocation in shared communication links, (ii) external interference from other devices or intentional jamming common in industrial environments, (iii) routing and protocol issues such as rooting loops, misconfigurations, or suboptimal routing decisions, (iv)  queuing delays and buffer limitations, e.g., due to bottlenecks at network nodes leading to packet drops when buffers are full,  and (v) physical damage to cables, faults in network interfaces, or other hardware malfunctions in routers and switches.
In case of link failure important information such as gradients or model updates, is not successfully delivered from one node to another. This can have a significant impact on the performance, feasibility, and convergence of the distributed optimization algorithm.

Next, we consider possible link failure (e.g., due to packet drops \cite{icrom}) over network $\mc{G}$. Recall that from Assumption~\ref{ass_W} our proposed solution in  Algorithm~\ref{alg_1} converges over uniformly-connected networks. In other words, the network can be possibly time-varying and even disconnected at some time instants but uniformly connected over certain time intervals. This is in contrast to existing ADMM-based solutions \cite{falsone2020tracking,cdc22,cdc_wei,banjac2019decentralized,notarnicola2022passivity,li2018admm,wang2021admm,carli2019distributed} that require all-time connectivity of network $\mc{G}$. Another useful feature of the proposed solution is that it only requires a balanced network (i.e., symmetric $W$) instead of \textit{stochastic} adjacency matrix $W$. Note that the weight-stochasticity of other existing solutions implies that compensation algorithms need to be designed to recover the loss of stochasticity due to the failure/removal of some links. Some redesign algorithms are proposed in the literature, see \cite{fagnani2009average,vaidya2012robust}. Therefore, the existing stochastic resource allocation strategies (e.g., \cite{xiao2006optimal,cherukuri2015distributed,falsone2020tracking,cdc22,cdc_wei,banjac2019decentralized,notarnicola2022passivity,li2018admm,wang2021admm,carli2019distributed}) may not properly work in case of link failure. In the absence of proper weight compensation algorithms the generation-demand balancing constraint would be violated, causing service disruption in the AGC setup.

\begin{thm} \label{thm_converg_drop}
	Under Assumptions~\ref{ass_W} and feasible initial conditions,  the solution dynamics \eqref{eq_sol2} (or \eqref{eq_sol}) converges to the optimizer $\mb{x}^*$ in the presence of link removal/failure if the network $\mc{G}$ remains uniformly connected.
\end{thm}
\begin{proof}
	First, note that all-time feasibility still holds under link removal, since $W$ remains symmetric. Recall that following the proof of Theorem~\ref{thm_converg2} and Lemma~\ref{lem_sum}, we need to prove that $\dot{\overline{F}} \leq 0$ over some time-interval $B$. Let $B$ denotes the uniform-connectivity time-interval under link failure, i.e., $\mc{G}_B=\cup_{t}^{t+B} \mc{G}(t)$ is connected with associated irreducible weight matrix $W^B$. Then, Eq.~\eqref{eq_proof2} can be redefined for the new B-connected network  $\mc{G}_B$\footnote{Note that the resulting dynamics~\eqref{eq_sol} in the presence of switching network topologies are time-varying and non-smooth.  Therefore, similar to the proof of Theorem~\ref{thm_converg2}, tools from non-smooth dynamic analysis \cite{cortes2008discontinuous} are used to prove convergence.}.
	
	\begin{align} \nonumber
		\dot{\overline{F}}
		= \sum_{i,j =1}^n \frac{\eta R  W^B_{ij}}{2W^B_{max}} &(\nabla f_j(x_j) - \nabla f_i(x_i))   \\ &\mbox{sat}\Big(g\Big(\nabla f_i(x_i) -  \nabla f_j(x_j)\Big)\Big)  \label{eq_proof3}
	\end{align} \normalsize
	which implies that $\dot{\overline{F}} \leq 0$. This holds if, at every time-instant, there is at least one link connecting (at least) two generator nodes over $\mc{G}_B$.
\end{proof}

Eq.~\eqref{eq_proof3} implies that the solution under link removal converges over a longer time interval associated with time $B$. In other words, following the notion of algebraic connectivity, convergence occurs at a slower rate.

	\begin{rem}
		The underlying network topology considerably affects the convergence rate and communication overhead of distributed AGC.
		Dense graphs, where each node communicates with many neighbours, typically enable faster information diffusion and, thus, a faster convergence rate. On the other hand, sparse graphs reduce the communication overhead but may slow down convergence and increase sensitivity to link failures. In this regard, the diameter of the network (maximum shortest path length between any two nodes) affects the time needed for information to propagate globally. Low diameter (high connectivity) implies faster distributed coordination since state updates incorporate more global information at each iteration. Moreover, redundant paths in the network topology may improve resilience against link failures.
	\end{rem}

\subsection{Convergence in the Presence of Communication Delay}\label{sec_delay}
Latency is a common issue in real-world communication networks and multi-agent systems. This section discusses distributed solution for AGC under communication delays.
The time delays are, in general, assumed heterogeneous at different links and at different time instants $k$. The time delay at link $(i,j)$ is denoted by $\tau_{ij}(k)\leq \overline{\tau}$ and, in the most general form, is arbitrary, random, and bounded by max possible delay $\overline{\tau}$. 
The delays are denoted by an indicator function $\mc{I}$ defined as
\begin{align} \label{eq_mcI}
	\mc{I}_{k,ij}(\tau) = \left\{ \begin{array}{ll}
		1, & \text{if}~  \tau_{ij}(k) = \tau,\\
		0, & \text{otherwise}.
	\end{array}\right.
\end{align}
It is assumed that the data sent from node $i$ to node $j$ is \textit{time-stamped}; this implies that the delay is known to the node $j$. Then, reformulating Eq.~\eqref{eq_sol} and \eqref{eq_sol_d}, the state-update at node $i$ is as follows
\begin{align} \nonumber
	{x}_i&(k+1) = {x}_i(k) \\
	&- {\eta_\tau} \frac{R}{W_{max}} \sum_{j \in \mc{N}_i} \sum_{r=0}^{\overline{\tau}} W_{ij} \mbox{sat} \Big(\nabla f_i(x_i(k-r)) -  \nabla f_j(x_j(k-r))\Big)\mc{I}_{k-r,ij}(r)
	\label{eq_sol_delay}
\end{align}
or
\begin{align} 
	\dot{x}_i = - \frac{\eta_\tau R}{W_{max}} \sum_{j \in \mc{N}_i} \sum_{r=0}^{\overline{\tau}} W_{ij} \mbox{sat} \Big(\nabla f_i(x_i(k-r)) -  \nabla f_j(x_j(k-r))\Big)\mc{I}_{k-r,ij}(r)
	\label{eq_sol_delay2}
\end{align}
with $\eta_\tau$ as the step-rate in the presence of communication delays.
Note that the delayed data exchange intuitively implies changes in the network structure. For example, if data shared between nodes $i,j$ is delayed from time-step $k$ to time-step $k+\tau$, one can simply remove the link $(i,j)$ from $\mc{G}(k)$ and add it in $\mc{G}(k+\tau)$. Therefore, one can extend the properties of the main dynamics \eqref{eq_sol_d} to the delayed case \eqref{eq_sol_delay} by considering changes in network topology $\mc{G}$ based on the time-delays. Recall that the delays are bounded by $\overline{\tau}$ and, therefore, $\mc{G}(k)$ associated with the time-delays is uniformly-connected over the time-period $B+\overline{\tau}$.
This implies that if in the delay-free case the network $\cup_{k}^{k+B} \mc{G}(k)$ is connected (uniform-connectivity over $B$ time-steps), then one can claim the connectivity of $\mc{G}_{B+\overline{\tau}}:=\cup_{k}^{k+B+\overline{\tau}} \mc{G}(k)$ in the presence of time-delays bounded by $\overline{\tau}$. Recall that the time delays in their general form are arbitrary, possibly time-varying, and heterogeneous over different links.
The following theorem states the convergence in the presence of delays.
\begin{thm} \label{thm_delay}
	Starting from feasible initial condition and under Assumptions~\ref{ass_W},  dynamics \eqref{eq_sol_delay} or \eqref{eq_sol_delay2}  in the presence of heterogeneous time-delays bounded by $\overline{\tau}$ converges to the optimal solution of  $\mb{x}^*$ for sufficiently small $\eta_\tau$.
\end{thm}
\begin{proof}
	The proof follows a similar procedure as in the proof of Theorem~\ref{thm_converg_drop} over the uniformly-connected network $\mc{G}_{T+\overline{\tau}}(k)$.
	Note that, following the definition of the union graph, the weight matrix of $\mc{G}_{B+\overline{\tau}}$, denoted by $W^{B+\overline{\tau}}$,  follows the same weight matrix of network $\mc{G}_{B}$, denoted by $W^{B}$ in Eq.~\eqref{eq_proof3}. This is because the same links are scattered over a longer time scale $B+\overline{\tau}$. Therefore, the proof follows similar to the proof of Theorem~\ref{thm_converg_drop} by considering longer time-scale $B+\overline{\tau}$ instead of $B$. Then, the rate of the Lyapunov function in Eq. \eqref{eq_proof3} can be reformulated for the new graph structure $\mc{G}_{B+\overline{\tau}}$ as
		\begin{align} \nonumber
		\dot{\overline{F}}
		= \sum_{i,j =1}^n \frac{\eta_\tau R  W^{B+\overline{\tau}}_{ij}}{2W^{B+\overline{\tau}}_{max}} &(\nabla f_j(x_j) - \nabla f_i(x_i))   \\ &\mbox{sat}\Big(g\Big(\nabla f_i(x_i) -  \nabla f_j(x_j)\Big)\Big)  \label{eq_proof4}
	\end{align} 
	which implies that $\dot{\overline{F}} \leq 0$ for every time-delay bounded by $\overline{\tau}$.  This completes the proof.
\end{proof}
Eq.~\eqref{eq_proof4} implies that the solution under communication time-delay converges over a longer time interval associated with $\overline{\tau}$ and, therefore, the rate of convergence is slower as the decrease in the Lyapunov function is over a longer time period depending on $\overline{\tau}$.

	\subsection{Comparison with the Literature}
	Here, we compare our proposed solution based on gradient-Laplacian tracking with some state-of-the-art literature. Due to addressing RRLs via saturation function the convergence rate of the proposed algorithm is slower than existing solutions including ADMM-based techniques \cite{falsone2020tracking,banjac2019decentralized,carli2019distributed,cdc22,cdc_wei,li2018admm,wang2021admm,he2019optimizing} and gradient-Laplacian solutions \cite{cherukuri2015distributed,chen2016distributed}. This is because our proposed dynamics limits the rate of power-update to address the RRLs while other existing methods may have faster rate of power-update without addressing the RRL constraints. Further, the proposed solution works over uniformly-connected networks while ADMM-based solutions require all-time connectivity. In terms of delay tolerance, most existing works including most ADMM-based solutions are susceptible to communication delay while our proposed technique handles heterogeneous time-varying bounded delays as discussed in Section~\ref{sec_delay}. Moreover, in contrast to ADMM-based techniques \cite{falsone2020tracking,banjac2019decentralized,carli2019distributed,cdc22,cdc_wei,li2018admm,wang2021admm,he2019optimizing} and gradient-Laplacian solutions \cite{cherukuri2015distributed,chen2016distributed}, the proposed gradient-tracking-based algorithm can handle nonlinear mappings on the links/nodes, e.g., due to data quantization or extra signum-based nonlinearity to accelerate convergence (as discussed in Section~\ref{sec_signum}). The comparison in this section is summarized in Table~\ref{tab_comp}.
	
	\begin{table} 
		\centering
			\caption{Comparison of the main properties of the existing algorithms} \label{tab_comp}	
			\scalebox{0.7}{\begin{tabular}{|c|c|c|c|c|c|}
				\hline
				Algorithm & Feasibility & RRLs & Connectivity & Delay & Nonlinearity\\
				\hline
				This work & All-Time & $\checkmark$ & Uniform &  Varying & $\checkmark$\\
				\hline
				Gradient-Laplacian \cite{cherukuri2015distributed,chen2016distributed} & All-Time & $\times$ & All-Time & $\times$ & $\times$\\
				\hline
				ADMM \cite{falsone2020tracking,banjac2019decentralized,carli2019distributed,cdc_wei,li2018admm,wang2021admm,he2019optimizing} & Asymptotic & $\times$ & All-Time & $\times$ & $\times$\\
				\hline
				ADMM \cite{cdc22} & Asymptotic & $\times$ & All-Time & Invariant & $\times$\\
				\hline
				\hline		
		\end{tabular}}
	\end{table} 

\section{Simulations} \label{sec_sim}
\subsection{Solution subject to RRLs: Comparison with the Literature} \label{sec_sim0}
In this section, we simulate protocol \eqref{eq_sol} and \eqref{eq_sol_d} to check whether the solution meets the constraints (both box constraint and RRL) and all-time feasibility. We also compare the solution with some existing algorithms in the literature. Consider a cyclic network (with symmetric $0$-$1$ weight matrix $W$) of $n=10$ generators. Every generator has two neighbours over this network, therefore $W_{max}=2$. The type of the generators is randomly chosen from Table~\ref{tab_par}. The problem is to compensate $P_{mis} = 700 MW$ while minimizing the associated cost functions from Table~\ref{tab_par}. The box constraints $\underline{m}_i$ and $\overline{m}_i$ are also given in Section~\ref{sec_prob}. The RRLs are considered equal to $R=1 MW/sec$ that needs to be met by all generators along the time-evolution of the solution. To address the box constraints smooth penalty terms in Eq.~\eqref{eq_f}-\eqref{eq_sigma} are considered with $c=1$ and $\sigma=2$. The states are initially chosen such that the sum equals the power mismatch. The easiest way to do so is to set all the initial conditions equal to $\frac{P_{mis}}{n}=70 MW$. The time-evolution of power states and power rates over $200$ steps and $\eta =1$ are shown in  Fig.~\ref{fig_sim1}.
\begin{figure}[]
	\centering
	\includegraphics[width=3.2in]{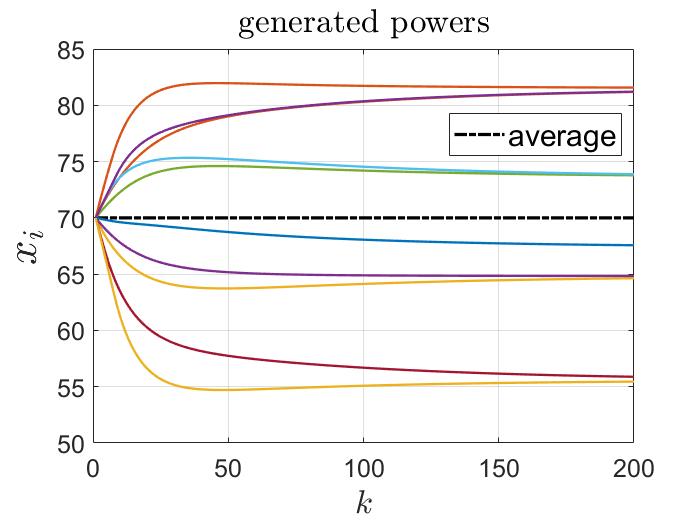}
	\includegraphics[width=3.2in]{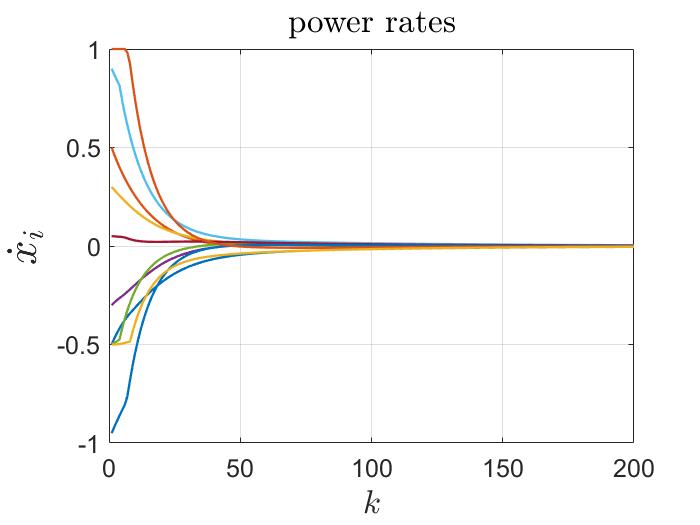}
	\caption{The solution under saturated protocol \eqref{eq_sol_d} to meet the RRLs and the box constraints. The feasibility constraint holds as the average of the power states is constant. The $\pm1$ RRLs are met by the power rates as shown in the right figure.
	}
	\label{fig_sim1}
\end{figure}
The generation cost residual (compared to the optimal cost) is shown in Fig.~\ref{fig_sim12} and is compared with some existing literature including linear  \cite{cherukuri2015distributed}, sign-based  \cite{ecc22}, finite-time \cite{chen2016distributed}, ADMM \cite{cdc_wei}, and DTAC-ADMM \cite{cdc22}  solutions. The cost function converges slower than other solutions which is due to the saturation nature of the proposed solution to address RRLs. The power rate at generator $1$ is shown as an example to check the RRLs. As it is clear from the figure, only the proposed solution handles the RRLs while the other existing solutions violate these RRL constraints and therefore only our solution properly works in practice. The computational complexity (run time) of the proposed solution is compared with the existing literature in Table~\ref{tab_comp2}.
	\begin{table} 
	\centering
	\caption{Comparison of computational complexity (run time in second) of the algorithms} \label{tab_comp2}	
		\begin{tabular}{|c|c|c|c|c|c|c|}
				\hline
				Algorithm & This work & Linear & Sign & Finite & ADMM & DTAC-ADMM\\
				\hline
				Run time & 0.0175 & 0.0171 & 0.0178 & 0.0188 & 0.0162 & 0.0165 \\
				\hline
				\hline		
	\end{tabular}	
\end{table} 
\begin{figure}[]
	\centering
	\includegraphics[width=3.2in]{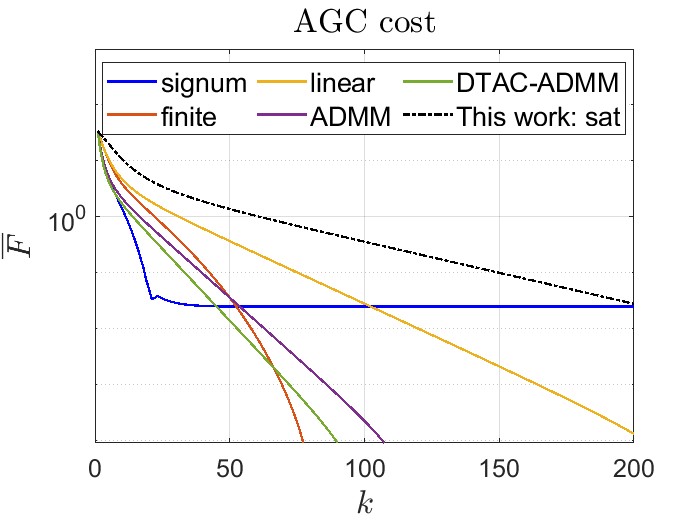}
	\includegraphics[width=3.2in]{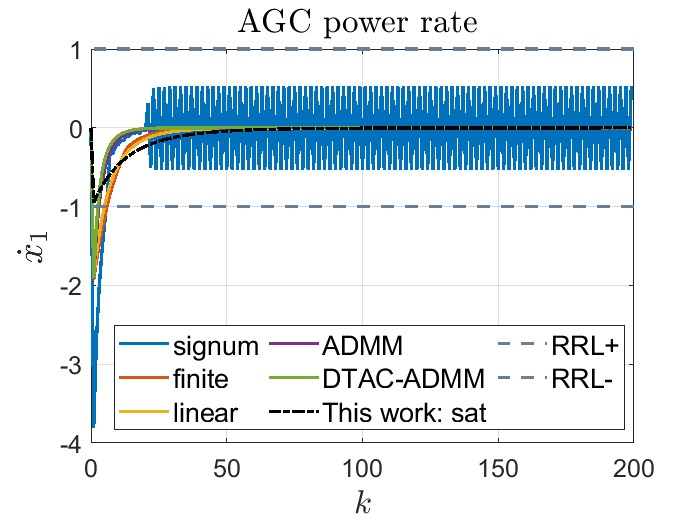}
	\caption{(Top) The cost evolution under solution \eqref{eq_sol_d} compared with some existing literature. Due to RRL limits and the saturation-type of the protocol the convergence is slow. (Bottom) The power rate at one sample generator is shown as an example. The RRL constraints are met by the solution while other solutions violate the RRLs and may not work properly in practice.
	}
	\label{fig_sim12}
\end{figure}

\subsection{Faster Convergence with Signum-based Nonlinearity} \label{sec_signum}
To improve the convergence rate signum-based nonlinearity is added to the protocol as discussed in Section~\ref{subsec_sign}. We repeat the same simulation (with the same parameters) as in the previous subsection. Fig.~\ref{fig_sim2} presents the time-evolution of the power states and generation rates under the protocol~\eqref{eq_sol_d2}.  For the signum nonlinearity, we consider $\mu = 0.6$. As compared to Fig.~\ref{fig_sim1} the generators faster reach the steady-state power generation while the rates still meet the RRLs. The generation-demand feasibility is still preserved by the new protocol as the average of power states is constant equal to $70$, implying that the generated powers at all times meet the $P_{mis} = 700 MW$.
\begin{figure}[]
	\centering
	\includegraphics[width=3.2in]{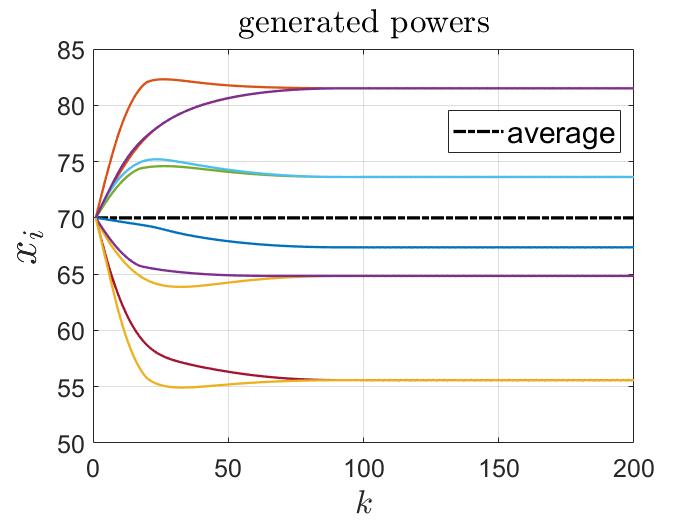}
	\includegraphics[width=3.2in]{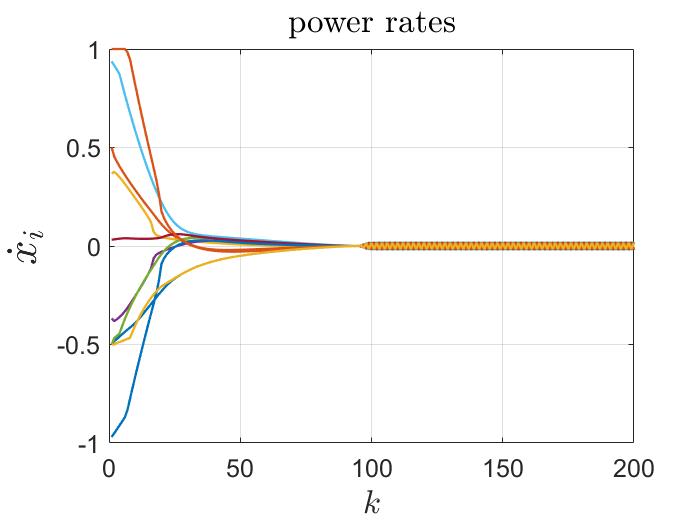}
	\caption{The solution under sat plus sgn protocol \eqref{eq_sol_d2} to converge faster to the steady-state optimal power states. The feasibility constraint still holds as the average of the power states is constant. The $\pm1$ RRLs are still met by the improved solution as shown in the right figure.
	}
	\label{fig_sim2}
\end{figure}
The time evolutions of the cost residuals are compared in Fig.~\ref{fig_sim22}. As it is clear from the figure, the convergence rate is significantly improved by the new protocol while the AGC power rates still meet the RRLs at the given sample generator.

\begin{figure}[]
	\centering
	\includegraphics[width=3.2in]{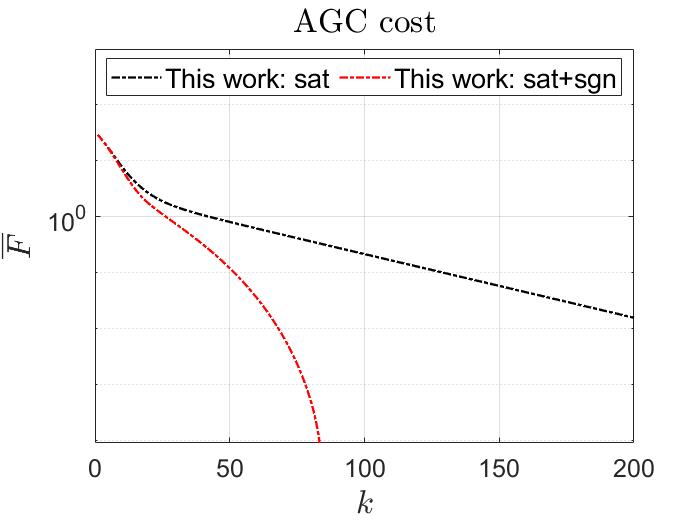}
	\includegraphics[width=3.2in]{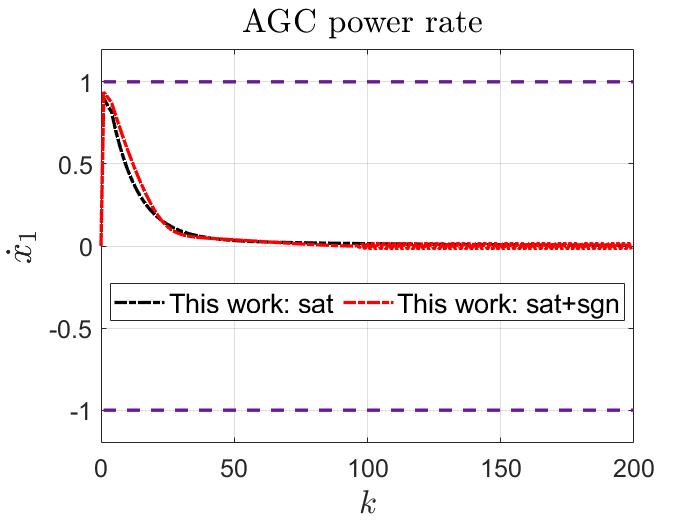}
	\caption{(Top) The cost evolution converges faster by adding sign-based nonlinearity. (Bottom) The power rate still meets the RRL constraint and works properly in practice as compared with the saturation-based solution.
	}
	\label{fig_sim22}
\end{figure}

\subsection{Link Failure Simulation}
Next, consider the AGC over a random Erdos-Renyi (ER) network with linking probability of $p=40\%$.  The motivation behind considering ER networks is that real-world communication networks in distributed power systems, such as wireless sensor networks or peer-to-peer communication among distributed generators, often exhibit randomness due to variable channel availability, packet drops, or link failures. ER networks serve as a natural graph model of such network setups where links randomly appear or disappear\footnote{In general, ER networks provide a standard baseline that is widely used in distributed algorithm research in the literature.}. These models well-capture the stochastic and unreliable nature of practical communication networks. Further, By varying the linking probability $p$, ER networks can represent a broad range of network densities from sparse to dense topologies.  
 Assume that, due to link failure, the network connectivity at some iterations reduces to $p=20\%$, $p=10\%$, and $p=5\%$.
We simulate a switching topology among the generators where at every $3$ iterations the network changes among these $4$ cases such that at every $B=12$ iterations $\mc{G}_B$ remains uniformly connected. The other simulation parameters are the same as in the previous subsections. The power states and the generation costs under dynamics \eqref{eq_sol_d} and \eqref{eq_sol_d2} are shown in Fig.~\ref{fig_switch}. Note that the network of $n=10$ generators under $5\%$ connectivity is disconnected, but the union network is B-connected.
\begin{figure}[hbpt]
	\centering
	\includegraphics[width=3.2in]{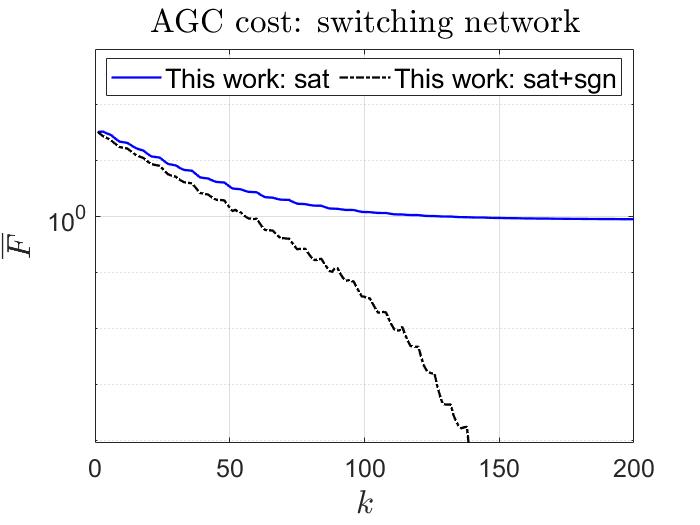}
	\includegraphics[width=3.2in]{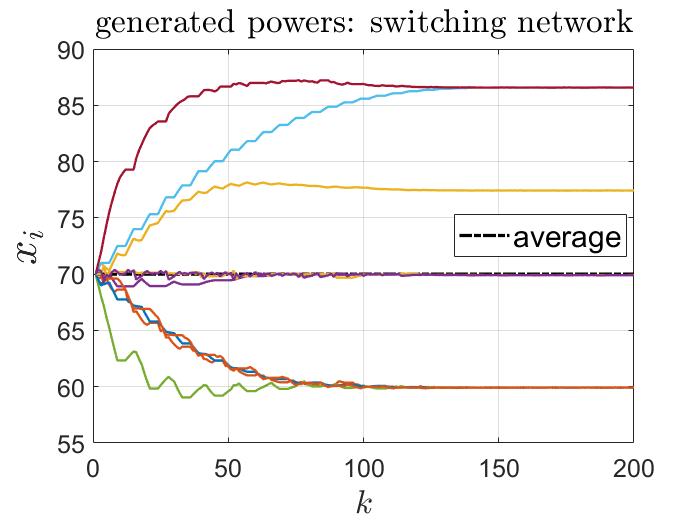}
	\caption{The residual cost evolution (Top) and the evolution of power states and their average (Bottom) under link failure and switching network topology.
	}
	\label{fig_switch}
\end{figure}

\subsection{Communication Delay Simulation}
	Next, consider the AGC over a random Erdos-Renyi network with linking probability of $40\%$ where the links are subject to bounded time-delays with $\overline{\tau}=4,8,12,16$ time-steps. The delays are considered random in the range $0 \leq \tau \leq \overline{\tau}$, time-varying, and heterogeneous at different links. The step-rate is set as $\eta_\tau = 0.2$. The other simulation parameters are the same as in the previous subsections. The power states and the generation costs under dynamics \eqref{eq_sol_delay}  are shown in Fig.~\ref{fig_delay}. It is clear from the figures that the solution is resilient to bounded time-delay for the given step-rate while the solution converges at a slower rate for larger time-delays. 
	\begin{figure}[hbpt]
		\centering
		\includegraphics[width=3.2in]{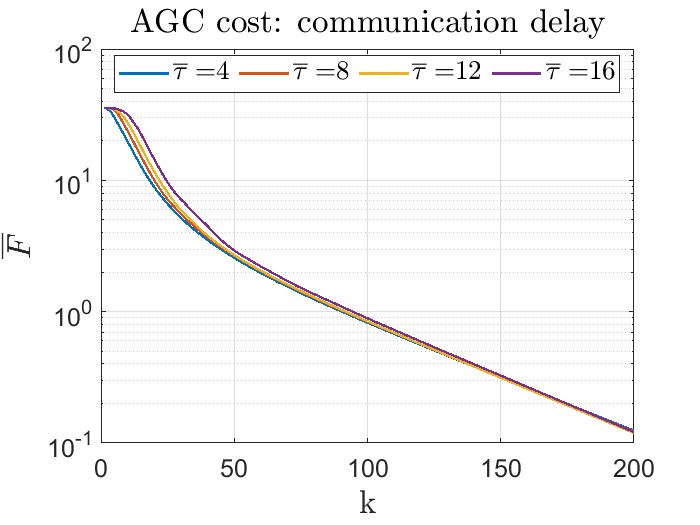}
		\includegraphics[width=3.2in]{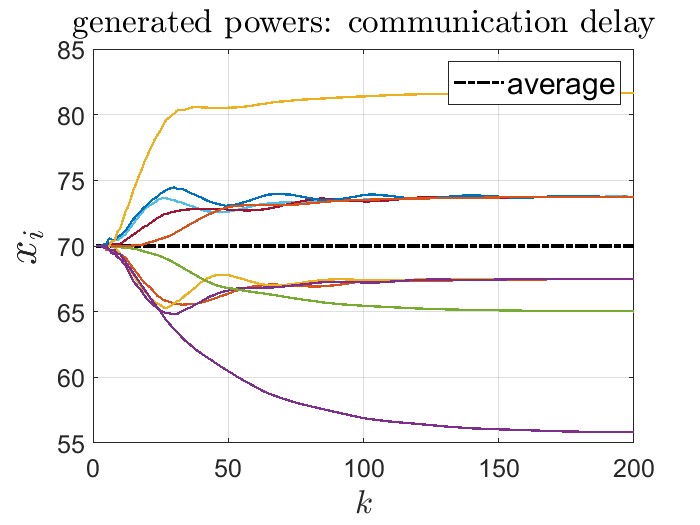}
		\caption{The residual cost evolution (Top) and the evolution of power states for $\overline{\tau}=16$  (Bottom) under communication time-delay bounded by $\overline{\tau}$ time-steps.
		}
		\label{fig_delay}
	\end{figure}

\subsection{Illustrating Scalability: Simulation over Large Networks}
	Next, to show the scalability of the proposed solution, we consider the AGC over a random Erdos-Renyi network of $n=200$ generator nodes with linking probability of $20\%$. The type of the generators is randomly chosen from Table~\ref{tab_par}. The problem is to compensate $P_{mis} = 14000 MW$ while minimizing the associated cost functions from Table~\ref{tab_par}. The RRLs are set equal to $R=1 MW/sec$ as in the previous simulations.
	The other simulation parameters are the same as in Section~\ref{sec_sim0}. The power states, power rates, and the generation cost under dynamics \eqref{eq_sol} over this large network are shown in Fig.~\ref{fig_large}. 
	\begin{figure}[hbpt]
		\centering
		\includegraphics[width=3.2in]{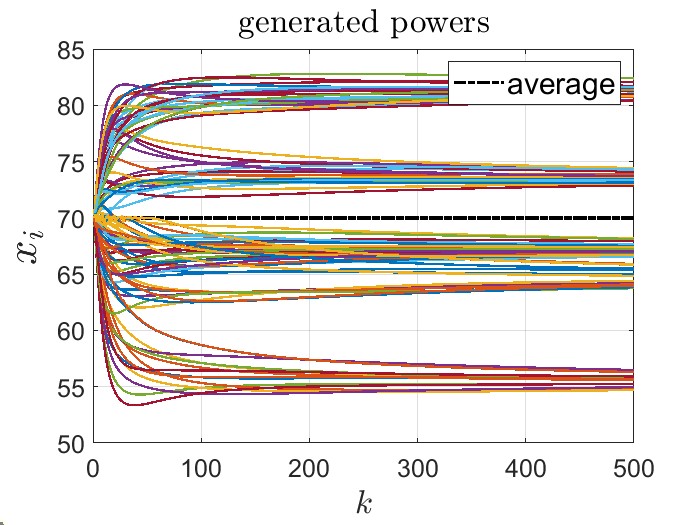}
		\includegraphics[width=3.2in]{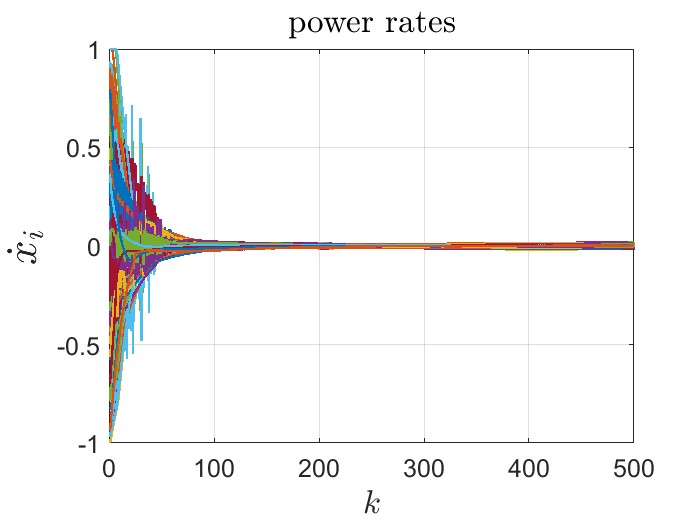}
		\includegraphics[width=3.2in]{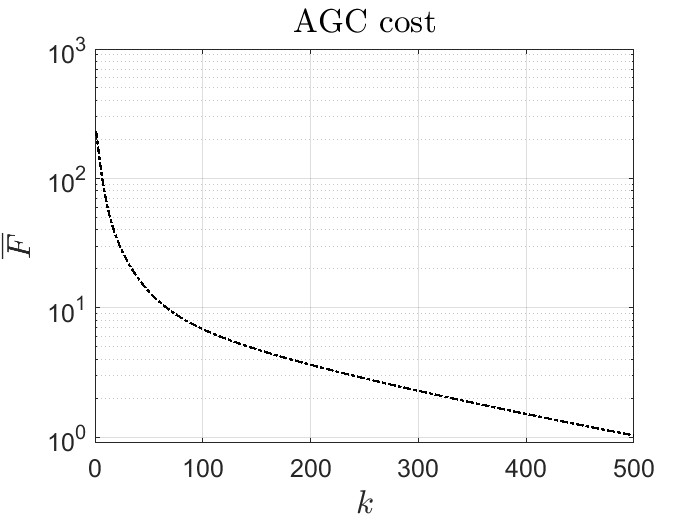}
		\caption{This figure shows the scalability of the solution under the proposed protocol \eqref{eq_sol} over large-scale networks while meeting the RRLs and the box constraints. 
		}
		\label{fig_large}
	\end{figure}

\section{Conclusions and Future Works} \label{sec_con}
\subsection{Concluding Remarks}
We consider a gradient-Laplacian solution for AGC resource allocation to address generators' RRLs and all-time feasibility. These cannot be addressed by the existing ADMM-based solutions. The RRLs are also not addressed by the existing linear/nonlinear gradient-Laplacian solutions. Without addressing the RRLs, the real generators cannot follow the rate assigned by the algorithm; therefore, the existing solutions cannot properly work in practice, i.e., either the feasibility constraint or the optimality is violated. To improve the cost-reduction rate, we also proposed additive sign-based nonlinearity to reach faster convergence.

\subsection{Future Directions}
As a future research direction, one can add signum-based consensus-type protocols, as discussed in consensus literature \cite{stankovic2020nonlinear}, to improve the noise-resiliency and disturbance-robustness.
Privacy-preserving solutions using the existing average consensus protocols \cite{mo2016privacy} is also another interesting future research direction. 

\section*{Acknowledgement}
This work is funded by Semnan University, research grant No. 226/1403/1403208.

\bibliographystyle{elsarticle-num} 
\bibliography{bibliography}

\end{document}